\documentclass[11pt]{article}
\usepackage{amsmath,amsbsy,amsfonts,amssymb,amsthm,color,dsfont,mleftright}

\def\ddefloop#1{\ifx\ddefloop#1\else\ddef{#1}\expandafter\ddefloop\fi}

\def\ddef#1{\expandafter\def\csname b#1\endcsname{\ensuremath{\mathbf{#1}}}}
\ddefloop ABCDEFGHIJKLMNOPQRSTUVWXYZ\ddefloop

\def\ddef#1{\expandafter\def\csname bb#1\endcsname{\ensuremath{\mathbb{#1}}}}
\ddefloop ABCDEFGHIJKLMNOPQRSTUVWXYZ\ddefloop

\def\ddef#1{\expandafter\def\csname c#1\endcsname{\ensuremath{\mathcal{#1}}}}
\ddefloop ABCDEFGHIJKLMNOPQRSTUVWXYZ\ddefloop

\def\ddef#1{\expandafter\def\csname v#1\endcsname{\ensuremath{\boldsymbol{#1}}}}
\ddefloop ABCDEFGHIJKLMNOPQRSTUVWXYZabcdefghijklmnopqrstuvwxyz\ddefloop

\def\ddef#1{\expandafter\def\csname v#1\endcsname{\ensuremath{\boldsymbol{\csname #1\endcsname}}}}
\ddefloop {alpha}{beta}{gamma}{delta}{epsilon}{varepsilon}{zeta}{eta}{theta}{vartheta}{iota}{kappa}{lambda}{mu}{nu}{xi}{pi}{varpi}{rho}{varrho}{sigma}{varsigma}{tau}{upsilon}{phi}{varphi}{chi}{psi}{omega}{Gamma}{Delta}{Theta}{Lambda}{Xi}{Pi}{Sigma}{varSigma}{Upsilon}{Phi}{Psi}{Omega}\ddefloop

\newcommand\veps{\ensuremath{\varepsilon}}

\renewcommand\t{{\ensuremath{\scriptscriptstyle{\top}}}}

\DeclareMathOperator{\tr}{tr}

\newcommand\wh{\ensuremath{\widehat}}
\renewcommand\v{\ensuremath{\boldsymbol}}

\newtheorem{claim}{Claim}

\newtheorem{theorem}{Theorem}

\theoremstyle{remark}
\newtheorem{remark}{Remark}
\theoremstyle{definition}

\newcommand\parens[1]{(#1)}
\newcommand\norm[1]{\|#1\|}

\newcommand\Parens[1]{\mleft(#1\mright)}
\newcommand\Norm[1]{\mleft\|#1\mright\|}
\newcommand\Braces[1]{\mleft\{#1\mright\}}

\usepackage[margin=1in,letterpaper]{geometry}
\usepackage[round]{natbib}

\newcommand\sr{\ensuremath{\operatorname{sr}}}

\title{Weighted sampling of outer products}
\author{Daniel Hsu\footnote{Department of Computer Science, Columbia
University. E-mail: \texttt{djhsu@cs.columbia.edu}}}

\begin{document}
\maketitle

\begin{abstract}
  This note gives a simple analysis of the randomized approximation
  scheme for matrix multiplication of~\citet{DKM06} with a particular
  sampling distribution over outer products.
  The result follows from a matrix version of Bernstein's inequality.
  To approximate the matrix product $\vA\vB^\t$ to spectral norm error
  $\veps\norm{\vA}\norm{\vB}$, it suffices to sample on the order of
  $(\sr(\vA) \vee \sr(\vB)) \log(\sr(\vA) \wedge \sr(\vB)) / \veps^2$
  outer products, where $\sr(\vM)$ is the stable rank of a matrix
  $\vM$.
\end{abstract}

\section{Introduction}

Let the matrices $\vA := [\va_1|\va_2|\dotsc|\va_n]$ and $\vB :=
[\vb_1|\vb_2|\dotsc|\vb_n]$ be given.
We are interested in developing an estimator of $\vA\vB^\t$
constructed as a positive linear combination of $m$ randomly chosen
outer products of the form $\va_i\vb_i^\t$.
We'll choose these outer products randomly (with replacement) using a
particular sampling distribution based on the lengths of the vectors
$\va_i$ and $\vb_i$.

A certain sampling distribution proposed by~\citet{DKM06} is optimal
for Frobenius norm error.
To obtain a spectral norm error of $\veps\norm{\vA}\norm{\vB}$, this
approach seems to require
\[
  m \gtrsim
  \frac{\sr(\vA)\sr(\vB) \log(\sr(\vA)\sr(\vB))}{\veps^2}
\]
samples, where
\[
  \sr(\vM) := \frac{\norm{\vM}_F^2}{\norm{\vM}^2}
\]
is the \emph{stable rank} of a matrix $\vM$.
Note that $\sr(\vM)$ is always at most the rank of $\vM$.

Another scheme proposed by~\citet{Sarlos06} first multiplies both
matrices $\vA$ and $\vB$ (on the right) by a random orthogonal matrix,
and then uses uniform random sampling to pick outer products formed
using columns of these randomly rotated matrices.
As shown by~\citet{HKZ12-mult}, for certain classes of random
orthogonal matrices, this scheme requires
\[
  m \gtrsim
  \frac{\Parens{\parens{\sr(\vA)\vee\sr(\vB)} + \log(n)}
  \log\parens{\sr(\vA)\vee\sr(\vB)}} {\veps^2}
\]
samples to obtain a spectral norm error bound of $\veps
\norm{\vA}\norm{\vB}$.%
\footnote{We use the notation $a \vee b = \max\{a,b\}$ and $a \wedge b
= \min\{a,b\}$.}

In this note, we describe a sampling distribution that requires
\[
  m \gtrsim
  \frac{\parens{\sr(\vA)\vee\sr(\vB)} \log\parens{ \sr(\vA) \wedge
  \sr(\vB)}}{\veps^2}
\]
samples to obtain a spectral norm error bound $\veps
\norm{\vA}\norm{\vB}$.
The sampling distribution is very natural; the purpose of this note is
merely to record a simple analysis based on a probability tail
inequality for sums of random matrices.

\section{The sampling scheme and estimator}

Define
\[
  p_i := \frac12 \Parens{
    \frac{\norm{\va_i}^2}{\norm{\vA}_F^2}
    + \frac{\norm{\vb_i}^2}{\norm{\vB}_F^2}
  }
  , \quad \forall i \in [n] .
\]
Observe that $p_i \geq 0$ for all $i \in [n]$, and $\sum_{i=1}^n p_i =
1$; thus $(p_1,p_2,\dotsc,p_n)$ is a valid probability distribution.
This sampling distribution is similar to one
from~\citet{BJS14-lowrank} for a different but related problem.
The proposed estimator has the same form as that of~\citet{DKM06}---it
is the empirical average of $m$ i.i.d.~random matrices $\vX_1, \vX_2,
\dotsc, \vX_m$:
\[
  \wh{\vA\vB^\t} := \frac1m \sum_{i=1}^m \vX_i
\]
where
\[
  \Pr\Braces{
    \vX_1 = \frac1{p_i} \va_i\vb_i^\t
  } = p_i
  , \quad \forall i \in [n] .
\]
(We assume without loss of generality that $p_i > 0$ for all $i \in
[n]$.)

For comparison, the sampling distribution of~\citeauthor{DKM06} has
$p_i \propto \norm{\va_i}\norm{\vb_i}$ but otherwise is the same.

\section{The result}

\begin{theorem}
  \label{thm:main}
  For any $t>0$,
  \begin{multline*}
    \Pr\Braces{
      \frac{\Norm{\wh{\vA\vB^\t} - \vA\vB^\t}}{\norm{\vA}\norm{\vB}}
      >
      \Parens{
        \sqrt{\frac{4\parens{ \sr(\vA) \vee \sr(\vB) } t}{m}}
        + \frac{\parens{\sqrt{\sr(\vA)\sr(\vB)} + 1} t}{m}
      }
    } \\
    \leq 4\parens{ \sr(\vA) \wedge \sr(\vB) }
    \cdot \frac{t}{e^t-t-1}
    .
  \end{multline*}
\end{theorem}
\begin{remark}
  In personal communication, John Holodnak and Ilse Ipsen informed me
  that they have also obtained essentially the same bound using this
  sampling scheme.
\end{remark}
\begin{remark}
  Theorem~\ref{thm:main} implies the following.
  There is a constant $c>0$ such that if
  \[
    m \geq c \cdot
    \Parens{
      \frac{
        \sr(\vA) \vee \sr(\vB)
      }{\veps^2}
      +
      \frac{
        \sqrt{\sr(\vA) \sr(\vB)}
      }{\veps}
    }
    \log\parens{ \sr(\vA) \wedge \sr(\vB) }
  \]
  for $\veps \in (0,1)$, then with high probability,
  \[
    \Norm{\wh{\vA\vB^\t} - \vA\vB^\t} \leq \veps\norm{\vA}\norm{\vB}
    .
  \]
\end{remark}
\begin{remark}
  The bound has no explicit dependence on the extrinsic dimensions of
  the $\va_i$ or $\vb_i$.
  We obtain this result by using a version of the matrix Bernstein
  inequality from~\citet{HKZ12-matrix} that depends only on intrinsic
  dimensions (which in this case are $\sr(\vA)$ and $\sr(\vB)$).
  It is straightforward to apply more recent versions of this
  inequality, such as one by \citet{Minsker12}, to obtain somewhat
  sharper probability tails.
\end{remark}

To prove Theorem~\ref{thm:main}, we shall apply a tail inequality for
the spectral norm of sums of symmetric random matrices
from~\citet{HKZ12-matrix}.
Define the symmetric random matrices
\[
  \vZ_i :=
  \begin{bmatrix}
    \v0 & \vX_i - \vA\vB^\t \\
    \vX_i^\t - \vB\vA^\t & \v0
  \end{bmatrix}
  , \quad \forall i \in [n] .
\]
Then
\[
  \Norm{\wh{\vA\vB^\t} - \vA\vB^\t}
  =
  \Norm{\frac1m \sum_{i=1}^m \vZ_i}
  .
\]
We need bounds on the following quantities: $\norm{\vZ_1}$ (with
probability one), $\norm{\bbE\vZ_1^2}$, and $\tr(\bbE\vZ_1^2)$.

\begin{claim}
  \label{claim:asbound}
  With probability one,
  $\norm{\vZ_1} \leq \norm{\vA}_F \norm{\vB}_F + \norm{\vA}\norm{\vB}$.
\end{claim}
\begin{proof}
  Observe that
  \begin{align*}
    p_i
    & = \frac12 \Parens{
      \frac{\norm{\va_i}^2}{\norm{\vA}_F^2}
      + \frac{\norm{\vb_i}^2}{\norm{\vB}_F^2}
    }
    \\
    & = \frac{
      \norm{\vB}_F^2\norm{\va_i}^2
      + \norm{\vA}_F^2\norm{\vb_i}^2
    }{
      2\norm{\vA}_F^2\norm{\vB}_F^2
    }
    \\
    & \geq \frac{
      2\norm{\vB}_F\norm{\va_i}
      \norm{\vA}_F\norm{\vb_i}
    }{
      2\norm{\vA}_F^2\norm{\vB}_F^2
    }
    \\
    & =
    \frac{
      \norm{\va_i}\norm{\vb_i}
    }{
      \norm{\vA}_F\norm{\vB}_F
    }
  \end{align*}
  for all $i \in [n]$, where the inequality relates arithmetic and
  geometric means.
  Therefore
  \begin{align*}
    \Norm{\frac1{p_i} \va_i\vb_i^\t - \vA\vB^\t}
    & \leq
    \frac1{p_i} \norm{\va_i\vb_i^\t} + \norm{\vA\vB^\t}
    \\
    & \leq
    \frac1{p_i} \norm{\va_i}\norm{\vb_i} + \norm{\vA}\norm{\vB}
    \\
    & \leq
    \norm{\vA}_F\norm{\vB}_F + \norm{\vA}\norm{\vB}
    .
  \end{align*}
  This means that $\norm{\vZ_1} \leq \norm{\vA}_F\norm{\vB}_F +
  \norm{\vA}\norm{\vB}$ with probability one.
\end{proof}

\begin{claim}
  \label{claim:variance}
  $\norm{\bbE\vZ_1^2} \leq 2\parens{ \sr(\vA) \vee \sr(\vB) }
  \norm{\vA}^2\norm{\vB}^2$.
\end{claim}
\begin{proof}
  First, it is easy to see that $\bbE\vX_1 = \vA\vB^\t$.

  Now observe that $\vZ_1^2$ is symmetric positive semidefinite, and
  \begin{align*}
    \bbE\vZ_1^2
    & =
    \begin{bmatrix}
      \sum_{i=1}^n \frac1{p_i} \norm{\vb_i}^2 \va_i\va_i^\t -
      \vA\vB^\t\vB\vA^\t & \v0 \\
      \v0 & \sum_{i=1}^n \frac1{p_i} \norm{\va_i}^2 \vb_i\vb_i^\t -
      \vB\vA^\t\vA\vB^\t
    \end{bmatrix}
    \\
    & \preceq
    \begin{bmatrix}
      \sum_{i=1}^n \frac1{p_i} \norm{\vb_i}^2 \va_i\va_i^\t & \v0 \\
      \v0 & \sum_{i=1}^n \frac1{p_i} \norm{\va_i}^2 \vb_i\vb_i^\t
    \end{bmatrix}
    .
  \end{align*}
  Therefore
  \[
    \norm{\bbE\vZ_1^2}
    \leq \Parens{
      \lambda_{\max}\Parens{\sum_{i=1}^n \frac1{p_i} \norm{\vb_i}^2 \va_i\va_i^\t}
      \vee
      \lambda_{\max}\Parens{\sum_{i=1}^n \frac1{p_i} \norm{\va_i}^2 \vb_i\vb_i^\t}
    }
    .
  \]
  We have that
  \[
    \sum_{i=1}^n \frac1{p_i} \norm{\vb_i}^2 \va_i\va_i^\t
    \preceq \max\Braces{ \frac1{p_i} \norm{\vb_i}^2 : i \in [n] }
    \Parens{\sum_{i=1}^n \va_i\va_i^\t}
    \preceq 2\norm{\vB}_F^2 \vA\vA^\t
  \]
  since $p_i \geq 0.5 \norm{\vb_i}^2/\norm{\vB}_F^2$ for each $i \in
  [n]$.
  Therefore
  \[
    \Norm{\sum_{i=1}^n \frac1{p_i} \norm{\vb_i}^2 \va_i\va_i^\t}
    \leq 2\norm{\vB}_F^2 \norm{\vA}^2 .
  \]
  Similarly,
  \[
    \Norm{\sum_{i=1}^n \frac1{p_i} \norm{\va_i}^2 \vb_i\vb_i^\t}
    \leq 2\norm{\vA}_F^2 \norm{\vB}^2 .
  \]
  This means that
  \begin{align*}
    \norm{\bbE\vZ_1^2}
    & \leq 2\parens{
      \norm{\vB}_F^2 \norm{\vA}^2 \vee
      \norm{\vA}_F^2 \norm{\vB}^2
    }
    \\
    & = 2\parens{ \sr(\vA) \vee \sr(\vB) } \norm{\vA}^2\norm{\vB}^2
    .
    \qedhere
  \end{align*}
\end{proof}

\begin{claim}
  \label{claim:trace}
  $\tr(\bbE\vZ_1^2) \leq 4\sr(\vA)\sr(\vB)\norm{\vA}^2\norm{\vB}^2$.
\end{claim}
\begin{proof}
  Using the expression for $\bbE\vZ_1^2$ from
  Claim~\ref{claim:variance}, we observe that
  \begin{align*}
    \tr(\bbE\vZ_1^2)
    & = 2 \sum_{i=1}^n \frac1{p_i} \norm{\va_i}^2 \norm{\vb_i}^2
    - 2\tr(\vA\vB^\t\vB\vA^\t)
    \\
    & \leq 2 \sum_{i=1}^n \frac1{p_i} \norm{\va_i}^2 \norm{\vb_i}^2
    \\
    & = 2 \sum_{i=1}^n \frac{2\norm{\va_i}^2\norm{\vb_i}^2}
    {
      \frac{\norm{\va_i}^2}{\norm{\vA}_F^2}
      + \frac{\norm{\vb_i}^2}{\norm{\vB}_F^2}
    }
    \\
    & = 2 \sum_{i=1}^n \frac{2}
    {
      \frac{1}{\norm{\vA}_F^2\norm{\vb_i}^2}
      + \frac{1}{\norm{\vB}_F^2\norm{\va_i}^2}
    }
    \\
    & \leq 4 \sum_{i=1}^n \norm{\vA}_F^2\norm{\vb_i}^2
    \\
    & = 4 \norm{\vA}_F^2 \norm{\vB}_F^2
    \\
    & = 4 \sr(\vA) \sr(\vB) \norm{\vA}^2 \norm{\vB}^2
    .
    \qedhere
  \end{align*}
\end{proof}

\begin{proof}[Proof of Theorem~\ref{thm:main}]
  We apply the matrix Bernstein inequality from~\citet{HKZ12-matrix}.
  We have from Claims~\ref{claim:asbound},~\ref{claim:variance},
  and~\ref{claim:trace},
  \begin{align*}
    \norm{\vZ_i}
    & \leq \norm{\vA}_F \norm{\vB}_F + \norm{\vA}\norm{\vB} =: \bar{b}
    \quad \text{(with probability one)} , \\
    \norm{\bbE\vZ_i^2}
    & \leq 2\parens{ \sr(\vA) \vee \sr(\vB) } \norm{\vA}^2\norm{\vB}^2
    =: \bar\sigma^2 , \\
    \tr(\bbE\vZ_1^2)
    & \leq \bar\sigma^2\bar{k}
  \end{align*}
  for
  \[
    \bar{k}
    := \frac{4\sr(\vA)\sr(\vB)\norm{\vA}^2\norm{\vB}^2}
    {\bar\sigma^2}
    = \frac{4\sr(\vA)\sr(\vB)\norm{\vA}^2\norm{\vB}^2}
    {2\parens{ \sr(\vA) \vee \sr(\vB) } \norm{\vA}^2\norm{\vB}^2}
    = 2\parens{\sr(\vA)\wedge\sr(\vB)}
    .
  \]
  Therefore, by the matrix Bernstein inequality from
  \citet{HKZ12-matrix} and a union bound,
  \[
    \Pr\Braces{
      \Norm{
        \frac1m \sum_{i=1}^m \vZ_i
      } > \sqrt{\frac{2\bar\sigma^2t}{m}} + \frac{\bar{b}t}{3m}
    } \leq  2\bar{k} \cdot \frac{t}{e^t-t-1}
    .
    \qedhere
  \]
\end{proof}

\bibliography{bib}
\bibliographystyle{plainnat}

\end{document}